\documentclass[journal]{IEEEtran} 				% For final IEEE paper version
\def\ps@headings{
\def\@oddhead{\mbox{}\scriptsize\rightmark \hfil \thepage}
\def\@evenhead{\scriptsize\thepage \hfil \leftmark\mbox{}}
\def\@oddfoot{}
\def\@evenfoot{}}
\usepackage{amsmath,amssymb}
\usepackage{amsthm}
\usepackage{color}
\usepackage[mathscr]{eucal}
\usepackage{graphics,graphicx,multicol}
\usepackage{epsfig}
\usepackage{enumerate}
\usepackage{subfigure}
\usepackage{algorithmic}
\usepackage[section]{algorithm}
\usepackage{morefloats}
\usepackage{enumerate}
\usepackage{subfigure}
\usepackage{multirow}
\usepackage{array}
\usepackage{pstricks, pst-node, pst-plot, pst-circ}
\usepackage{moredefs}
\usepackage{multirow}
\usepackage{cite}
\usepackage{courier}
\newtheorem{thm}{Theorem}[section]
\newtheorem{lem}[thm]{Lemma}

\newtheorem{pro}[thm]{Proposition}
\begin{document}
\title{Optimal Radio Resource Allocation for Hybrid Traffic in Cellular Networks: Traffic Analysis and Implementation}
\author{Mo~Ghorbanzadeh,~\IEEEmembership{}
        Ahmed~Abdelhadi,~\IEEEmembership{}
        Charles~Clancy~\IEEEmembership{}
\thanks{M. Ghorbanzadeh, A. Abdelhadi, and C. Clancy are with the Hume Center for National Security and Technology, Virginia Tech, Arlington,
VA, 22203 USA e-mail: \{mgh, aabdelhadi,tcc\}@vt.edu.

Part of this work was accepted at IEEE ICNC CNC Workshop 2015 \cite{GhorbanzadehCNC2015_1}.}
}

\maketitle

\begin{abstract}
In part I of this paper, a distributed and a centralized architecture for an optimal radio resource allocation aware of the traffic delay-tolerance nature, user subscription type, and application usage variations were developed. In the current article, a transmission overhead analysis of the aforementioned distributed and a centralized architectures is investigated and it is proved that the centralized scheme endures a significantly lower transmission overhead than does the distributed approach. Furthermore, the lower bounds of the transmission overhead for both the centralized and the distributed architectures are derived. Moreover, a sensitivity analysis of the resource allocation procedures of the aforesaid centralized and distributed architectures to the changes in the number of users in the system is presented. Besides, a sensitivity analysis of the centralized and distributed approaches to the temporal changes in application usages are investigated. Ultimately, the transmission overhead and sensitivity 
relevant statements are verified through appropriate simulations. And last but not the least, a real-world implementation of the resource allocation methods developed in Part I is provided.
\end{abstract}


\providelength{\AxesLineWidth}       \setlength{\AxesLineWidth}{0.5pt}%
\providelength{\plotwidth}           \setlength{\plotwidth}{8cm}% width of the axes only
\providelength{\LineWidth}           \setlength{\LineWidth}{0.7pt}%
\providelength{\MarkerSize}          \setlength{\MarkerSize}{3pt}%
\newrgbcolor{GridColor}{0.8 0.8 0.8}%
\newrgbcolor{GridColor2}{0.5 0.5 0.5}%

\section{Introduction}\label{sec:intro}
Resource allocation methods aiming at fulfilling the quality of service (QoS) requirements of current-day cellular networks have been the focus of many research studies. In fact, cellular communication systems nowadays contain smartphones capable of running several applications simultaneously. Since the applications have miscellaneous QoS criteria based on the type of the traffic that they deal with, many of the resource allocation techniques have to account for the application types in order to incorporate QoS into their operation. Furthermore, the diversity of the users has to be included in the resource allocation methods implemented in networks catering to a wide variety of subscribers. Besides, concurrent running of applications on the smartphones motivates incorporating the application usage temporal changes into the resource allocation schemes. Much of the state of the art in single carrier resource allocation (\cite{Lee05non-convexoptimization,AbdelhadiCNC2014, AbdelhadiPIMRC2013, 
AbdelhadiMobicom2013}) considers either traffic or subscriber types into their formulation. However, some of these work fail to provide optimality for the allocated rates or do not account for all the aforementioned QoS-related issues at the same time. Many other multi carrier resource allocation studies (\cite{ShajaiahICNC2014,ShajaiahMILCOM2013,ShajaiahPIMRC2014,ShajaiahCCS2014}) also have been done, which do not address all the aforesaid QoS-related concerns.

A single carrier optimal radio resource allocation with utility proportional fairness formulation was presented in part I \cite{GhorbanzadehPart2} of the article at hand. The aforementioned formulation was cast into a centralized and a distributed architecture in the part I of this paper \cite{GhorbanzadehPart2}, and also included algorithms to solve them. The current article is part II of the resource allocation paper presented in \cite{GhorbanzadehPart2} and tries to address the QoS-related issues mentioned above with regard to the method presented in \cite{GhorbanzadehPart2}. Hereafter, a referring to "part II" indicates part II of the resource allocation paper \cite{GhorbanzadehPart2}, referred to as "part I", unless otherwise is explicitly stated. Part II first presents a traffic analysis of the centralized and distributed resource allocation architectures that were introduced in part I.

The traffic analysis includes a transmission overhead analysis of the aforesaid architectures and obtains lower bounds for the transmission overheads. Furthermore, part II contains a sensitivity analysis of the centralized and distributed architectures presented in the part I \cite{GhorbanzadehPart2} of the current article. The sensitivity analysis studies whether the user equipment (UE) bids and their allocated rates remain optimal when the number of UEs or their application usage percentages change in the system. The sensitivity analysis also investigates the variations of the transmission overhead of the centralized and distributed architectures, developed in part I \cite{GhorbanzadehPart2}, when the UE quantity or application usage percentages change and under the existence or lack of a rebidding procedure in response to the aforesaid changes. Then, simulations are done to confirm to-be-derived predicates relevant to the transmission overhead and sensitivity analyses. The traffic and sensitivity analysis 
deductions, applicable to part I, can be applied to the resource block allocation formulations in (\cite{Erpek2015,GhorbanzadehICNC2015}) as well as to the radar-spectrum shared resource allocation work in \cite{GhorbanzadehMILCOM2014}. Ultimately, a real-world network is developed on which the resource allocation algorithms in part I is implemented. The implementation results reveal that applying the algorithms improves the QoS for the applications. The current article is intended to complement the part I of the paper.

\subsection{Contributions}
The contributions in the current paper are summarized below.

\begin{itemize}
\item We analyze the sensitivity of the distributed and centralized resource allocations developed in part I to the variations in the number of UEs in the system and to changes in the applications usage percentages.
\item We investigate the transmission overhead associated with the distributed and centralized resource allocations proposed in Part I and obtain the lower bounds.
\item We provide with relevant simulations for the aforementioned sensitivity and transmission overhead analyses.
\item We implement the resource allocations developed in part I on a real-world network and show the effectiveness of the proposed method in QoS elevation.
\end{itemize}

\subsection{Organization}\label{sec:organization}
The remainder of the current paper proceeds as follows. Section \ref{sec:Background} presents the necessary background information from part I to facilitate understanding the part II without frequent referring to the part I. Section \ref{sec:QnDn} presents the traffic analysis of the centralized and distributed resource allocations to the UE quantity variations. Section \ref{sec:SensitivityUsage} portrays a traffic analysis of the proposed schemes to changes of the application usage percentages. Section \ref{sec:complexity} provides with a brief discussion on computational complexity of the centralized and distributed resource allocation methods developed in part I. Section \ref{sec:sim} presents the simulation results relevant to the transmission overhead and sensitivity of the proposed architectures. Ultimately, section \ref{Sec:Implementation} presents a real-world network implementation of the resource allocation; and, section \ref{sec:conclude} concludes the paper.

In the next section, we present the preliminary information, extracted from part I of this article,in an effort to minimize the readers' need to refer to part I in order to understand part II.

\section{Background}\label{sec:Background}
We consider a cellular communications systems with $M$ UEs served by an Evolved Node-B (eNB) in Figure \ref{fig:system_model}. The UEs run real-time and delay-tolerant applications which are represented by utility functions. The utility functions $U(r)$ represent applications QoS satisfaction in terms of the allocated rates $r$ and have the following properties \cite{Shenker95fundamentaldesign}.

\begin{figure}[!htb]
\begin{center}
\includegraphics[width=3.5in]{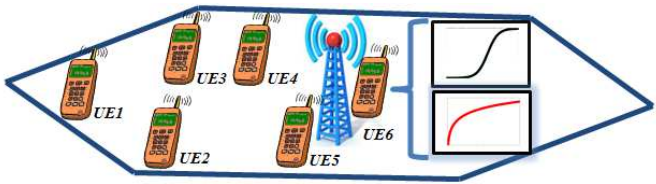}
\end{center}
\caption{\footnotesize{System Model: Single cell, within the cellular network, with an eNB covering $M = 6$ UEs each with simultaneously running delay-tolerant and relay-time applications represented by logarithmic and sigmoidal utility functions respectively.}}\label{fig:system_model}
\end{figure}

\begin{itemize}
\item $U(0) = 0$ and $U(r)$ is an increasing function $r$ .
\item $U(r)$ is twice differentiable in $r$ and bounded above.
\end{itemize}

The first statement of the the first predicate in the aforementioned properties implies that the utility functions are nonnegative. This is expected since utility functions represent application QoS satisfaction percentage. The second statement of the first predicate explains that when more rate is assigned to the applications, they obtain a higher QoS satisfaction. The second predicate indicates that the utility functions are continuous. The delay-tolerant and real-time applications are respectively modelled by normalized logarithmic and sigmoidal utility functions represented in equations (\ref{eqn:sigmoid}) and (\ref{eqn:log}) in that order \cite{DL_PowerAllocation}.

\begin{equation}\label{eqn:sigmoid}
    U(r) = c\Big(\frac{1}{1+e^{-a(r-b)}}-d\Big)
\end{equation}

Here, $c = \frac{1+e^{ab}}{e^{ab}}$ and $d = \frac{1}{1+e^{ab}}$. It can be easily verified that $U(0) = 0$ and $U(\infty) = 1$, where the former is one of the previously mentioned utility function properties and the latter indicates that assigning infinite resources ensues $100\%$ QoS satisfaction. Furthermore, it is easily derivable that the inflection point of the sigmoidal utility function in equation (\ref{eqn:sigmoid}) occurs at $r = r^{\text{inf}} = b$, where the superscript "inf" stands for infliction. In equation (\ref{eqn:log}), $r^{\text{max}}$ is the maximum rate at which the application QoS achieves $100\%$ utility satisfaction percentage and $k$ is the utility function increase with augmenting the allocated rate $r$. It can be easily verrified that $U(0) = 0$ and $U(r^{\text{max}}) = 1$., where the former is the property of the utility functions and the latter implies that $r = r^{max}$ achieves  a $100\%$ QoS satisfaction. Moreover, the inflection point of normalized logarithmic function is at 
$r = r^{\text{inf}} = 0$.

\begin{equation}\label{eqn:log}
    U(r) = \frac{\log(1+kr)}{\log(1+kr^{\text{max}})}
\end{equation}

The objective is to determine the optimal rates that the eNBs should be allocating to their UE applications to ensure that Real-time applications are rendered priority over delay-tolerant ones, no user is dropped, applications temporal usage variations are accounted for, and subscription-based treatments are honored. We assume that UEs run multiple real-time and delay-tolerant applications, mathematically represented by sigmoidal and logarithmic utility functions. To obtain the optimal rates and fulfill the above requirements, part I of the paper developed the resource allocation algorithms cast under a centralized architecture, which is displayed again for the readers' convenience.

\subsection{Centralized Resource Allocation}\label{sec:Centralized}
Here, the application rates are directly assigned by the eNBs in a singular stage. The formulation of the centralized algorithm is illustrated in equation (\ref{eqn:opt_multiapp}). Here, $M$ UEs are covered by an eNB and $\textbf{r} = [r_1,r_2,...,r_M]$ is the UE allocated rate vector, whose $i^{th}$ element is the rate allocated to the $i^{th}$ UE. Besides, $r_{ij}$, $U_{ij}(r_{ij})$, and $\alpha_{ij}$ respectively represent the rate allocation, application utility function, and application usage percentage of the $j^{th}$ application running on the $i^{th}$ UE. Hence, we can write $\sum_{j=1}^{N_i}{\alpha_{ij}} = 1$ which states that the addition of the usage percentage of the application running on the $i^{th}$ UE is $100\%$. Furthermore, we can write $r_i = \sum_{j=1}^{N_i}{r_{ij}}$ where $N_i$ is the number of coevally running applications on the $i^{th}$ UE. The former of the afore-written equations , and the letter one implies that the the $i^{th}$ UE rate is the augmentation of all its $N_i$ 
applications resources assignments.Also, $R$ is the maximum bandwidth available to the eNB and $\beta_i$ is a subscription-dependent weight for the $i^{th}$ UE. Regarding the centralized architecture, part I of this article does the following:

\begin{itemize}
\item Part I proves that the resource allocation in equation (\ref{eqn:opt_multiapp}) is convex and has a tractable solution.
\item Part I proves that the resource allocation in equation (\ref{eqn:opt_multiapp}) refrains from dropping UEs, prioritizes real-time applications, considers application usage variations, and accounts for UE priorities.
\item Part I provides with algorithms, implemented in the UEs and eNBs, to obtain the optimal application rates.
\end{itemize}

\begin{equation}\label{eqn:opt_multiapp}
\begin{aligned}
& \underset{\textbf{r}}{\text{max}}
& & \prod_{i=1}^{M}\Big(\prod_{j=1}^{N_i}U_{ij}^{\alpha_{ij}}(r_{ij})\Big)^{\beta_{i}} \\
& \text{subject to}
& & \sum_{i=1}^{M}\sum_{j=1}^{N_i}r_{ij} \leq R,\\
& & &  r_{ij} \geq 0, \;\; i = 1,2, ...,M,\;\;j = 1,2,...,N_i\\
\end{aligned}
\end{equation}

More details about the solution algorithms for the distributed architecture can be found in part I. In addition, part I of the paper developed the resource allocation algorithms cast under a distributed architecture, which is displayed below for the readers' convenience.

\subsection{Distributed Resource Allocation}\label{sec:Distributed}
In the distributed architecture, there are two optimization problems, the first of which concerns the UE optimal rates allocated by the eNBs via collaborations between the eNBs and the UEs that they are covering. This stage is referred to as external UE resource allocation (EURA) and is written as equation (\ref{eqn:opt_sub1}), where for $M$ UEs $V_i(r_i) = \prod_{j=1}^{N_i}U_{ij}^{\alpha_{ij}}(r_{ij})$ is the aggregated utility function for the $i^{th}$ UE, $\textbf{r} = [r_1,r_2,...,r_M]$ is the UE rate vector whose $i^{th}$ component represents the rate assigned by the eNB to the $i^{th}$ UE. The second optimization focuses on distributing the application rates by their host UEs and is performed internally to the UEs. This stage is denoted as the internal UE rate allocation (IURA) and is written in equation (\ref{eqn:opt_sub2}). Here, $\textbf{r}_i = [r_{i1},r_{i2},...,r_{iN_i}]$ is the application rate allocation vector such that its $j^{th}$ component indicates the bandwidth allotted by the $i^{th}$ UE 
to its $j^{th}$ application, $r_{i}^{\text{opt}}$ is the $i^{th}$ UE rate allocated by eNB via solving the EURA optimization in equation (\ref{eqn:opt_sub1}), and $N_i$ is the number of applications running in the $i^{th}$ UE. Superscript "opt" indicates the optimality of the UE rates which was proved in part I \cite{GhorbanzadehPart2}, which also proved that there exists a tractable global optimal solution to the IURA optimization problem in equation (\ref{eqn:opt_sub2}) and also provided the solving algorithm thereof. In summary, regarding the distributed architecture, part I of the article does the following:

\begin{itemize}
\item Part I proves that the EURA and IURA in equations (\ref{eqn:opt_sub1}) and (\ref{eqn:opt_sub2}) are convex and have tractable solutions.
\item Part I proves that the EURA and IURA in equations (\ref{eqn:opt_sub1}) and (\ref{eqn:opt_sub2}) refrain from dropping UEs, IURA prioritize real-time applications, IURA considers application usage variations, and EURA accounts for UE priorities.
\item Part I provides with algorithms, implemented in the UEs and eNBs, to obtain the optimal application rates.
\end{itemize}

\begin{equation}\label{eqn:opt_sub1}
\begin{aligned}
& \underset{\textbf{r}}{\text{max}}
& & \prod_{i=1}^{M}V_i^{\beta_{i}}(r_{i}) \\
& \text{subject to}
& & \sum_{i=1}^{M}r_{i} \leq R,\\
& & &  r_{i} \geq 0, \;\;\;\;\; i = 1,2, ...,M.
\end{aligned}
\end{equation}

\begin{equation}\label{eqn:opt_sub2}
\begin{aligned}
& \underset{\textbf{r}_i}{\text{max}}
& & \prod_{j=1}^{N_i}U_{ij}^{\alpha_{ij}}(r_{ij}) \\
& \text{subject to}
& & \sum_{j=1}^{N_i}r_{ij} \leq r_{i}^{\text{opt}},\\
& & & r_{ij} \geq 0, \;\;\;\;\; j = 1,2, ...,N_i.
\end{aligned}
\end{equation}

More details about the solution algorithms for the distributed architecture can be found in part I.

\section{Traffic Analysis under UE Quantity Dynamics}\label{sec:QnDn}
We consider a cell of $M$ UEs and an eNB as in part I. One iterate of the centralized resource allocation algorithm assigns optimal rates to applications from the eNB to UEs, whereas the distributed approach requires more iterations thereof. For the centralized scheme, the following proposition \ref{pro:overhead} is conceivable.

\begin{pro}\label{pro:overhead}
The centralized resource allocation assigns optimal rates with a minimum transmission overhead $2M$.
\end{pro}

\begin{proof}
Corollary III.4 in part I proved the optimality of the centralized algorithm, so the first predicament of the proposition \ref{pro:overhead} is substantiated as well. The centralized algorithm repetitions too assign an application rate is one; as such $M$ transmissions of the application utility parameters from $M$ UEs to their eNB proceeds with $M$ one-to-one transmissions of the optimal rates from the eNB to the $M$ UEs. Hence, the transmission overhead of the centralized approach is $2M$. This is also the minimum possible transmissions to permit initiating connections and assigning rates for the $M$ UEs in the system.
\end{proof}

Lets consider a scenario where the number of UEs changes from $M_1$ at time slot $n = n_1$ to $M_2$ $n = n_1+1$. Therefore, the instantaneous UE quantity $M(n)$ can be written as equation (\ref{eqn:UECg}).

\begin{equation}\label{eqn:UECg}
M(n) = \left\{
  \begin{array}{l l}
    M_1\:\: ; \:\:n\leq n_1\\
    M_2\:\: ; \:\:n>n_1
  \end{array} \right.
\end{equation}

We assume that steady state rates are reached at the time slots $n_1$ and $n_2$ for the $M_1$ and $M_2$ users respectively so that the steady state rate vector $\textbf{r}^{\text{ss}}$ can be expressed as equation (\ref{eqn:RtCg}).

\begin{equation}
\label{eqn:RtCg}
\textbf{r}^{\text{ss}} = \left\{
  \begin{array}{l l}
    \textbf{r}(n_1)\:\: ; \:\: n\leq n_1,\\
    \textbf{r}(n_2)\:\: ; \:\: n > n_1.
  \end{array} \right.
\end{equation}

Based on equations (\ref{eqn:UECg}) and (\ref{eqn:RtCg}), we present a traffic analysis of the distributed algorithm with rebidding in section \ref{sec:SensitivityRebid}.

\subsubsection{Distributed Algorithm with Rebidding}\label{sec:SensitivityRebid}
For every UEs entering/leaving the cell, we assume that prior UEs (i.e. users already in the cell and having been allotted optimal rates) will rebid for resources as the system has changed. In this case, proposition \ref{pro:new_user_rebidding_overhead} is considerable with regard to the transmission overhead.

\begin{pro}\label{pro:new_user_rebidding_overhead}
The transmission overhead of the distributed resource allocation algorithm with rebidding is maximum and greater than $2M_2+2-M_1$.
\end{pro}

\begin{proof}
When $M_2>M_1$, i.e. the UE quantity increases, new UEs send their initial bids requiring $M_2-M_1$ transmissions to the eNB. If $\beta_i$ is available at the UEs, the eNB broadcasts a new shadow price in $M_2-M_1+1$ transmissions. Next, the $M_2$ UEs transmit the eNB their new bids followed by the eNB sending another shadow price. This is repeated $k$ times until an optimal rate allocation occurs which amounts to the transmission overhead $(M_2-M_1)+1+kM_2+k = (k+1)M_2+k+1-M_1$, corresponding minimally to $2M_2+2-M_1$ for one iteration ($k=1$). On the other hand, if $\beta_i$ is only available at the eNB, it transmits a modified shadow price $\frac{p_E}{\beta_i}$ to the UEs under its coverage. Thus, the transmission overhead will be $(M_2-M_1)+M_2+2kM_2 = (2k+2)M_2-M_1$, corresponding minimally to $4M_2-M_1$ for one iteration ($k=1$). Hence, the transmission overhead is greater than $2M_2+2-M_1$.

When $M_1>M_2$, i.e. the UE quantity decreases, the exiting UEs send service termination signals to the eNB in $M_1-M_2$ transmissions. If $\beta_i$ is available at the UEs, the eNB broadcasts a new shadow price in $M_1-M_2 +1$ transmissions. Next, the $M_2$ UEs transmit the eNB their new bids followed by the eNB sending another shadow price. This process is repeated $k$ times until an optimal rate allocation occurs which amounts to the transmission overhead $(M_1-M_2)+1+kM_2+k = (k-1)M_2+k+1+M_1$, corresponding minimally to $2M_2+2-M_1$ for one iteration ($k=1$). On the other hand, if $\beta_i$ is only available at the eNB, it transmits a modified shadow price $\frac{p_E}{\beta_i}$ to the UEs under its coverage. Thus, the transmission overhead becomes $(M_1-M_2)+M_2+2kM_2 = 2kM_2+M_1$ transmissions, corresponding minimally to $2M_2+M_1$ for one iteration ($k=1$). Hence, the transmission overhead of the distributed approach with rebidding is no less than $2M_2+2-M_1$.
\end{proof}

For the distributed scheme with rebidding, the following proposition \ref{pro:new_user_rebidding} discusses a sensitivity of rates and bids to UE quantity variations.

\begin{pro}\label{pro:new_user_rebidding}
The allocated rates and pledged bids with rebidding remain optimal in the distributed resource algorithm under UE quantity changes.
\end{pro}

\begin{proof}
Without loss of generality, we let the number of UEs change from $M_1$ to $M_2$, denoted as $M_1 \rightarrow M_2$. Assuming the distributed algorithm arrives at the steady state according to the equation (\ref{eqn:RtCg}), theorem III.2 and corollary III.3 conclude that the EURA and IURA optimizations with $M = M_1$ are optimal so that for $M_1$ UEs before the time slot $n_1$, we have:

\begin{equation}
\begin{aligned}\label{eqn:n1_RtDst}
r_i(n_1) &= S_i^{-1}(p(n_1))\\
w_i(n_1) &= p(n_1)S_i^{-1}(p(n_1))
\end{aligned}
\end{equation}
where $p(n_1)=\frac{\sum_{i=1}^{M_1}w_i(n_1)}{R}$ and the allocations are done in accordance UE EURA, eNB EURA, and IURA algorithms with initial bids $w_i(1) = 1$.  Likewise, theorem III.2 and corollary III.3 conclude that the EURA and IURA optimizations with $M = M_2$ are optimal when the users rebid for resources under UE quantity changes so that for $M_2$ UEs before time slot $n_2$, we have:

\begin{equation}
\begin{aligned}\label{eqn:n2_time_slot}
r_i(n_2) &= S_i^{-1}(p(n_2)),\\
w_i(n_2) &= p(n_2)S_i^{-1}(p(n_2)),
\end{aligned}
\end{equation}
where $p(n_2)=\frac{\sum_{i=1}^{M_2}w_i(n_2)}{R}$ and the reallocations are done in accordance with UE EURA, eNB EURA, and IURA algorithms with initial bids $w_i(n_1)$. Therefore, the rates and bids remain optimal for the distributed scheme with rebidding once the UE quantity changes.
\end{proof}

Next, section \ref{sec:SensitivityNoRebid} analyzes the sensitivity of bid and rates of the distributed approach to UE quantity changes when no-rebidding happens.

\subsubsection{Distributed Algorithm without rebidding}\label{sec:SensitivityNoRebid}
For every UEs entering/leaving the cell, we assume that prior UEs (i.e. users already in the cell and having been allotted optimal rates) will not rebid for resources so that priori UEs bids do not vary from those in the allocation before time slot $n=n_1$, when UE quantity alters. It is shown in proposition \ref{pro:new_user_no_rebidding_overhead} that the transmission overhead is less than the situation where users rebid for resources.

\begin{pro}\label{pro:new_user_no_rebidding_overhead}
The transmission overhead for the distributed resource allocation with no rebidding is not larger than the one with rebidding and is $M_2+1-M_1$ for $M_2>M_1$ and $M_1-M_2$ for $M_1>M_2$ at minimum.
\end{pro}

\begin{proof}
When $M_2>M_1$, i.e. the UE quantity increases, new UEs send their initial bids requiring $M_2-M_1$ transmissions to the eNB which broadcasts a new shadow price only to the new UEs if $\beta_i$ is available at the UEs. This is repeated $k$ times until an optimal rate allocation is reached which amounts to the transmission overhead $k(M_2-M_1)+k$, corresponding minimally to $M_2-M_1+1$ for one iteration ($k=1$). On the other hand, if $\beta_i$ is only available at the eNB, it transmits a modified shadow price $\frac{p_E}{\beta_i}$ to the $M_2-M_1$ UEs in $2(M_2-M_1)+2k(M_2-M_1) = (2k+2)(M_2-M_1)$ transmissions, corresponding minimally to $4(M_2-M_1)$ for one iteration ($k=1$).

When $M_1>M_2$, i.e. the UE quantity decreases, exiting UEs send service termination signals to the eNB in $M_1-M_2$ transmissions. There are no further transmissions and the transmission overhead becomes $M_1-M_2$. Therefore, the minimum transmission overhead is $M_2+1-M_1$ for $M_2>M_1$ and $M_1-M_2$ for $M_1>M_2$, both less than the minimum $2M_2+2-M_1$ transmissions of the distributed scheme with rebidding proved in proposition \ref{pro:new_user_rebidding_overhead}.
\end{proof}

Next, lemma \ref{lem:new_user_no_rebidding} proves that the rates and bids are not optimal for the distributed scheme with no rebidding when the UE quantity changes.

\begin{lem}\label{lem:new_user_no_rebidding}
The allocated rates and bids without rebidding are not optimal for the distributed resource allocation under UE quantity changes.
\end{lem}

\begin{proof}
Assuming a rate and bid optimal convergence for the $M_1$ UEs according to equation (\ref{eqn:RtCg}) before the time slot $n_1$, we have:

\begin{equation}
\begin{aligned}\label{eqn:n1_time_slot}
r_i(n_1) &= S_i^{-1}(p(n_1)),\\
w_i(n_1) &= p(n_1)S_i^{-1}(p(n_1)),
\end{aligned}
\end{equation}
where $p(n_1)=\frac{\sum_{i=1}^{M_1}w_i(n_1)}{R}$ and the allocations are done in accordance with UE EURA, eNB EURA, and IURA algorithms with initial bids $w_i(1) = 1$.

When $M_2>M_1$, i.e. the UE quantity increases, prior users bids ($i \in \{1,2, ..., M_1\}$) remain fixed for $n>n_1$ and the shadow price can be written as:

\begin{equation}
\begin{aligned}\label{eqn:shadow_price_n1}
p(n>n_1) &= \frac{\sum_{i=1}^{M_1}w_i(n_1)+\sum_{i=M_1+1}^{M_2}w_i(n>n_1)}{R},\\
p(n>n_1) &= p(n_1) + \frac{\sum_{i=M_1+1}^{M_2}w_i(n>n_1)}{R}.
\end{aligned}
\end{equation}
and user bids can be expressed as equation (\ref{eqn:n2_time_slot_no_rebidding}).

\begin{equation}
\label{eqn:n2_time_slot_no_rebidding}
w_i(n>n_1) = \left\{
  \begin{array}{l l}
    p(n_1)S_i^{-1}(p(n_1))\:\: ; \:\:i=\{1,2,...,M_1\},\\
    p(n>n_1)S_i^{-1}(p(n>n_1))\:\: \\
    \:\:\:\:\:\:\:\:; \:\:i=\{M_1+1,M_1+2,...,M_2\}.
  \end{array} \right.
\end{equation}

Assuming a convergence for the $M_2$ UEs according to equation (\ref{eqn:RtCg}) before the time slot $n_2$, the optimal bids and rates can be written as equation (\ref{eqn:bids_n2_time_slot_no_rebidding}).

\begin{equation}
\begin{aligned}\label{eqn:bids_n2_time_slot_no_rebidding}
w_i(n_2) = \left\{
  \begin{array}{l l}
    p(n_1)S_i^{-1}(p(n_1))\:\: ; \:\:i=\{1,2,...,M_1\}\\
    p(n_2)S_i^{-1}(p(n_2))\:\: \\
    \:\:\:\:\:\:\:\:; \:\:i=\{M_1+1,M_1+2,...,M_2\}\\
  \end{array} \right.\\
r_i(n_2) = \left\{
  \begin{array}{l l}
    \frac{p(n_1)}{p(n_2)}S_i^{-1}(p(n_1))\:\: ; \:\:i=\{1,2,...,M_1\},\\
    S_i^{-1}(p(n_2))\:\: \\
    \:\:\:\:\:\:\:\:; \:\:i=\{M_1+1,M_1+2,...,M_2\}.
  \end{array} \right.
\end{aligned}
\end{equation}
where $p(n_1) = \frac{\sum_{i=1}^{M_1}w_i(n_1)}{R}$ and $p(n_2) = p(n_1) + \frac{\sum_{i=1+M_1}^{M_2}w_i(n_2)}{R}$ and the reallocation is done in accordance with UE EURA, eNB EURA, and IURA algorithms with the initial bids $w_i(n_1)$. Contrasting equations (\ref{eqn:bids_n2_time_slot_no_rebidding}) and (\ref{eqn:n2_time_slot}), we observe the rates and bids are non-optimal.
\end{proof}

Section \ref{sec:SensitivityCentralized} discusses the sensitivity of the centralized resource allocation to UE quantity alterations.

\subsubsection{Sensitivity of the Centralized Algorithm}\label{sec:SensitivityCentralized}
A transmission overhead analysis of the centralized resource allocation, assigning optimal rates after one algorithm iteration, is explained in proposition \ref{pro:2overhead}.

\begin{pro}\label{pro:2overhead}
The centralized resource allocation scheme assigns optimal rates with a minimum transmission overhead $2M_2-M_1$ and $M_1$ under increasing and decreasing UE quantities, respectively.
\end{pro}

\begin{proof}
When $M_2>M_1$, i.e. the UE quantity increases, the new UEs send their utility function parameters in $M_2-M_1$ transmissions to the eNB which returns application rates for the new $M_2$ UEs. Thus, the transmission overhead becomes $2M_2-M_1$. On the other hand, for $M_1>M_2$, i.e. the UE quantity decreases, the exiting UEs transmit service termination signals in $M_1-M_2$ transmissions to the eNB which sends the rates for the new $M_2$ UEs. There are no further transmissions and the transmission overhead becomes $M_1-M_2+M_2 = M_1$.
\end{proof}

Next, section \ref{sec:SensitivityUsage} presents a sensitivity analysis of the distributed and centralized resource allocations to temporal application usage changes.

\section{Traffic Analysis under Application Usage Dynamics}\label{sec:SensitivityUsage}
We consider a cell of $M$ UEs and an eNB as in part I. At time slot $n_1$, the $i^{th}$ user has the aggregated utility function $V_i(r_i)$. At time instant $n_1+1$, $M'$ of the $M$ UEs ($M' < M$) change their aggregated utility, e.g. the $i^{th}$ UE of the $M'$ users gets the aggregated utility function $V'_i(r_i)$. Thus, the time-dependent aggregated utility function for the $i^{th}$ UE of the $M'$ users, denoted as $V_i(r_i,n)$, can be written as equation (\ref{eqn:AgUtTime}).

\begin{equation}
\label{eqn:AgUtTime}
V_i(r_i,n) = \left\{
  \begin{array}{l l}
    V_i(r_i)\:\: ; \:\: n\leq n_1\\
    V'_i(r_i)\:\: ; \:\:n > n_1
  \end{array} \right.
\end{equation}

Furthermore, we assume that a steady state rate allocation is achieved for the aggregated utility functions $V_i(r_i)$ and $V'_i(r_i)$ at time slots $n_1$ and $n_2$ respectively so that the steady state rate vector can be expressed as equation (\ref{eqn:rates_n2_time_slot_no_rebiddingSteadyState}).

\begin{equation}
\label{eqn:rates_n2_time_slot_no_rebiddingSteadyState}
\textbf{r}^{\text{ss}} = \left\{
  \begin{array}{l l}
    \textbf{r}(n_1)\:\: ; \:\: n\leq n_1\\
    \textbf{r}(n_2)\:\: ; \:\: n > n_1
  \end{array} \right.
\end{equation}

Like section \ref{sec:QnDn}, in addition to the centralized scheme, we consider the distributed algorithms with and without rebidding, of which the former is discussed in section \ref{sec:DsRbUsDn}.

\subsubsection{Distributed Algorithm with Rebidding}\label{sec:DsRbUsDn}
For every UEs altering their application usage percentages, which in turn change the corresponding aggregated utility functions, fixed-utility users (UEs with unvaried aggregated utility and optimal rates) rebid for resources. The eNB initiates the allocation by either sending (in case $\beta_i$ is only available at the eNB) or broadcasting (in case $\beta_i$ is available at the UEs) the shadow price. Proposition \ref{pro:new_utility_rebidding_overhead} proves the transmission overhead maximum. Moreover, the rates and bids remain optimal as explained in the proposition \ref{pro:new_utility_rebidding}.

\begin{pro}\label{pro:new_utility_rebidding_overhead}
The transmission overhead of the distributed resource allocation with rebidding is maximum and greater than $M+M'+2$ under application usage changes.
\end{pro}

\begin{proof}
The $M'$ UEs changing their aggregated utility send their initial bids in $M'$ transmissions to the eNB. If $\beta_i$ is available at the UEs, the eNB broadcasts a new shadow price to the devices followed by $M$ UEs transmit new bids to the eNB. This process is repeated $k$ times until optimal rates are achieved and the transmission overhead becomes $M'+1+kM+k$, corresponding minimally $M+M'+2$ for one iteration ($k = 1$). On the other hand, if $\beta_i$ is only available at the eNB, it transmits a modified shadow price $\frac{p_E}{\beta_i}$ in $M$ transmissions to the UEs under its coverage. The $M$ UEs send new bids to which eNB and this process is iterated $k$ times until optimal rates are achieved. The transmission overhead becomes $M'+M+2kM$ minimally corresponding to $3M+M'$ for one iteration ($k = 1$). Contrasting minimal overheads in both cases, the minimum transmission overhead is $M+M'+2$.
\end{proof}

\begin{pro}\label{pro:new_utility_rebidding}
The allocated rates and bids of the distributed resource allocation with rebidding remain optimal under application usage changes.
\end{pro}

\begin{proof}
Theorem III.2 and corollary III.3 concludes that the new EURA and IURA optimizations are optimal for $M$ users. Assuming the aggregated UE utility $V_i(r_i)$ and slope curvature function $S_i(r_i)$, the optimal rates and bids at time slot $n_1$ can be expressed as equation set (\ref{eqn:n1_time_slot_utility}).

\begin{equation}
\begin{aligned}\label{eqn:n1_time_slot_utility1}
r_i(n_1) &= S_i^{-1}(p(n_1))\\
w_i(n_1) &= p(n_1)S_i^{-1}(p(n_1))
\end{aligned}
\end{equation}
where $p(n_1)=\frac{\sum_{i=1}^{M}w_i(n_1)}{R}$ and allocations are done in accordance with UE EURA, eNB EURA, and IURA algorithms with initial bids $w_i(1) = 1$. Once $M'$ UEs alter the usage percentage of at least one of their applications leading to their aggregated utility functions $V'_i(r_i)$ and slope curvature functions $S'_i(r_i)$, the $M-M'$ UEs keep their aggregated utilities $V_i(r_i)$ and slope curvature functions $S_i(r_i)$. Thus, optimal rates and bids at time slot $n_2$ can be expressed as:

\begin{equation}
\begin{aligned}\label{eqn:aaa}
w_i(n_2) = \left\{
  \begin{array}{l l}
    p(n_1)S_i^{'-1}(p(n_2))\:\: ; \:\:i=\{1,2,...,M'\}\\
    p(n_2)S_i^{-1}(p(n_2))\:\: \\
    \:\:\:\:\:\:\:\:; \:\:i=\{M'+1,M'+2,...,M\}\\
  \end{array} \right.\\
  \\
r_i(n_2) = \left\{
  \begin{array}{l l}
    \frac{p(n_1)}{p(n_2)}S_i^{-1}(p(n_1))\:\: ; \:\:i=\{1,2,...,M_1\},\\
    S_i^{'-1}(p(n_2))\:\: ; \:\:i=\{1,2,...,M'\} \\
    S_i^{-1}(p(n_2))\:\: ; \:\:i=\{M'+1,M'+2,...,M\}.
  \end{array} \right.
\end{aligned}
\end{equation}

where $p(n_2)=\frac{\sum_{i=1}^{M}w_i(n_2)}{R}$ and the reallocation procedure is done in accordance with UE EURA, eNB EURA, and IURA algorithms with initial bids $w_i(n_2)$. Contrasting equations (\ref{eqn:aaa}) and (4), we observe that the rates and bids are optimal under application usage changes.
\end{proof}

Next, section \ref{sec:DsnRbUsDn} investigates the sensitivity of the distributed algorithm with no rebidding to application usage alterations.

\subsubsection{Distributed Algorithm without Rebidding}\label{sec:DsnRbUsDn}
For every UEs altering their application usage percentages, which in turn change the corresponding aggregated utility functions, fixed-utility users (UEs with unvaried aggregated utility and optimal rates) will not rebid for resources so that their rates and bids at time $n_2$ are identical to those of time slot $n=n_1$. In this case, the transmission overhead is less than the one with rebidding as shown in proposition \ref{pro:new_utility_no_rebidding_overhead}. Besides, the rates and bids are not optimal as illustrated in lemma \ref{lem:new_utility_no_rebidding}.

\begin{pro}\label{pro:new_utility_no_rebidding_overhead}
The transmission overhead of the distributed resource allocation with no rebidding is not larger than the one with rebidding and equals $M'+1$ at minimum.
\end{pro}

\begin{proof}
$M'$ UEs changeing their aggregated utility functions send the initial bids in $M'$ transmissions to the eNB. If $\beta_i$ is available at the UEs, the eNB broadcasts a new shadow price to the UEs, the routine is repeated $k$ times until optimal rates are achieved, and the transmission overhead becomes $kM'+k$, corresponding minimally to $M'+1$ for one iteration ($k = 1$). On the other hand, if $\beta_i$ is only available at the eNB, it transmits a modified shadow price $\frac{p_E}{\beta_i}$ to the $M'$ UEs, the algorithm is iterated $k$ times until optimal rates are achieved, and the transmission overhead becomes $2M'+2kM' = (2k+2)M'$, corresponding minimally to $4M'$ for one iteration ($k = 1$). So the minimum transmission overhead for the distributed algorithm with no-rebidding becomes $M'+1$, which is less than overhead for the distributed scheme with rebidding (proposition \ref{pro:new_utility_rebidding_overhead}).
\end{proof}

\begin{lem}\label{lem:new_utility_no_rebidding}
The allocated rates and bids of the distributed resource allocation without rebidding are not optimal under application usage percentage changes.
\end{lem}

\begin{proof}
Assuming convergence before the time slot $n_1$, the optimal rates and bids at time slot $n_1$ can be expressed as equation set (\ref{eqn:n1_time_slot_utility}) where $p(n_1)=\frac{\sum_{i=1}^{M}w_i(n_1)}{R}$ and the allocation is done in accordance with UE EURA, eNB EURA, and IURA algorithms with initial bids $w_i(1) = 1$.

\begin{equation}
\begin{aligned}\label{eqn:n1_time_slot_utility}
r_i(n_1) &= S_i^{-1}(p(n_1))\\
w_i(n_1) &= p(n_1)S_i^{-1}(p(n_1)).
\end{aligned}
\end{equation}

The $i^{th}$ UE bids ($\{i|i=\{M'+1,M'+2, ..., M\}\}$) is fixed for time slots $n>n_1$. After applications usage changes, we can write the shadow price as equation set (\ref{eqn:shadow_price_n1_utility}).

\begin{equation}
\begin{aligned}\label{eqn:shadow_price_n1_utility}
p(n>n_1) &= \frac{\sum_{i=1}^{M'}w_i(n>n_1)+\sum_{i=1+M'}^{M}w_i(n_1)}{R}\\
p(n>n_1) &= \frac{\sum_{i=1}^{M'}w_i(n>n_1)}{R} + p(n_1).
\end{aligned}
\end{equation}

And user bids can be expressed as equation (\ref{eqn:n1_time_slot_no_rebidding_utility}).

\begin{equation}
\label{eqn:n1_time_slot_no_rebidding_utility}
w_i(n>n_1) = \left\{
  \begin{array}{l l}
    p(n>n_1)S_i^{'-1}(p(n>n_1))\:\: \\
    \:\:\:\:\:\:\:\:; \:\:i=\{1,2,...,M'\},\\
    p(n_1)S_i^{-1}(p(n_1))\:\: \\
    \:\:\:\:\:\:\:\:; \:\:i=\{M'+1,M'+2,...,M\}.
  \end{array} \right.
\end{equation}

Assuming convergence to optimal values before the time slot $n_2$, final rates and bids can be expressed as:

\begin{equation}
\begin{aligned}\label{eqn:n2_time_slot_no_rebidding_utility2}
w_i(n_2) = \left\{
  \begin{array}{l l}
    p(n_1)S_i^{'-1}(p(n_2))\:\: \\
    \:\:\:\:\:\:\:\:; \:\:i=\{1,2,...,M'\} \\
    p(n_1)S_i^{-1}(p(n_1))\:\: \\
    \:\:\:\:\:\:\:\:; \:\:i=\{M'+1,M'+2,...,M\}\\
  \end{array} \right.\\
r_i(n_2) = \left\{
  \begin{array}{l l}
    \frac{p(n_2)}{p(n_1)}S_i^{'-1}(p(n_2))\:\: \\
    \:\:\:\:\:\:\:\:; \:\:i=\{1,2,...,M'\},\\
    S_i^{-1}(p(n_1))\:\: \\
    \:\:\:\:\:\:\:\:; \:\:i=\{M'+1,M'+2,...,M\}
    \end{array} \right.
\end{aligned}
\end{equation}
where $p(n_1) = \frac{\sum_{i=1}^{M}w_i(n_1)}{R}$ and $p(n_2) = \frac{\sum_{i=1}^{M'}w_i(n_2)}{R} + p(n_1)$ and the procedure for the reallocation is done in accordance with UE EURA, eNB EURA, and IURA algorithms with initial bids $w_i(n_1)$. Contrasting the optimal rates in equation (\ref{eqn:n2_time_slot_no_rebidding_utility2}) and the rates in equation (4), we observe that the rates and bids for the distributed resource allocation without rebidding do not remain optimal under changing application usage percentages.
\end{proof}

Section \ref{sec:SensitivityCentralizedUsg} investigates the sensitivity of the centralized resource allocation to application usage alterations.

\subsubsection{Sensitivity of the Centralized Algorithm}\label{sec:SensitivityCentralizedUsg}
Here, we assume $M'$ out of $M$ ($M' < M$) UEs in the system change their application usages. For the centralized resource allocation scheme, optimal rates are assigned in one iteration, i.e. one transmission and one reception per UE. Proposition \ref{pro:central_utility_overhead_optimiality} analyzes the transmission overhead of the centralized approach to applications usage changes.

\begin{pro}\label{pro:central_utility_overhead_optimiality}
The centralized resource allocation algorithm assigns optimal rates with a minimum transmission overhead $M'+M$ when $M'$ UEs change their applications usage percentages.
\end{pro}

\begin{proof}
The $M'$ UEs changing their aggregated utilities send their utility parameters in $M'$ transmissions to the eNB, which transmits new rates to the $M$ UEs. Therefore, the transmission overhead is $M'+M$. The optimality of the centralized algorithm has been proved by the corollary III.4 in part I; and, the proof is complete.
\end{proof}

Next, section \ref{sec:complexity} discusses the computational complexity associated with the centralized resource allocation scheme presented by the Algorithms V.1 and V.2 with those of the distributed approach in the Algorithms IV.3, IV.4, and IV.5.

\section{Computational Complexity and Power Consumption Insights}\label{sec:complexity}
The centralized resource allocation undergoes less computational complexity on the UE side compared to the distributed one as it does not include the computations of UE EURA algorithm of the two-stage distributed scheme. As a result, the power consumption required for implementing the centralized scheme in the UE is less. Next, in the section \ref{sec:sim}, we present simulation results for the resource allocation methods developed in Part I.

\section{Simulation Results}\label{sec:sim}
The communication cell, in part I, with $M = 6$ UEs and an eNB is considered such that each UE concurrently runs a delay-tolerant and a real-time applications with utility parameters in Table I of part I, where we have mentioned that the real-time are voice-over-IP, standard video, and high definition video and the delay tolerant ones are FTPs. Next, section \ref{subsec:SensitivityAlpha} presents simulations germane to the sensitivity of the algorithms to the changes in the application usage percentage.

Figure \ref{fig:TnOvCnDs} depicts a comparison between the transmission overhead incurred employing the centralized and the distributed resource allocation approaches, between which the former is concomitant with so significantly less overhead that the germane plot is much closer to the abscissa as opposed to the one for the latter case. This is in light of the fact that the centralized scheme subsumes less message exchanges between the UEs and eNB (one stage only) to assign optimal rates to the applications.

\begin{figure}[t!]
\centering
  \includegraphics[width=\plotwidth]{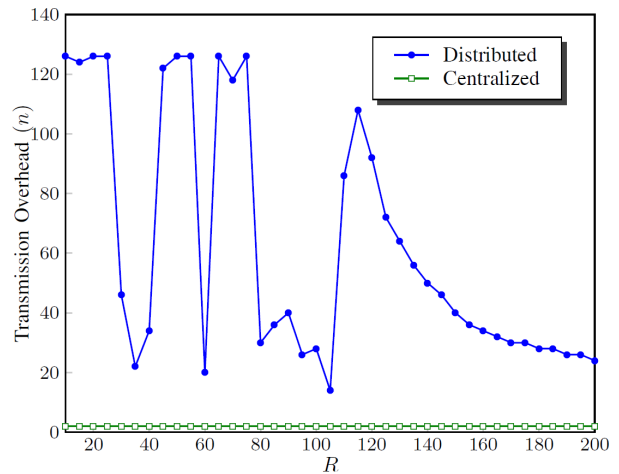}
  \caption{Transmission Overhead vs. eNB Rate $R$: The centralized resource allocation scheme incurs a drastically less transmission overhead as opposed to the distributed approach, as it requires less message exchanges before it converges to the optimal rates. This is in sight of the fact that the centralized method has the eNB assign application rates in one stage.}
  \label{fig:TnOvCnDs}
\end{figure}

Next, we show the sensitivity of the algorithms to temporal alterations in application usage percentages.

\subsection{Sensitivity to Application Usage Changes} \label{subsec:SensitivityAlpha}
In the following simulations, we set the termination threshold as $\delta = 10^{-3}$ and the total achievable rate at the eNB as $R = 180$. The, we measure the sensitivity of the algorithm to changes in the usage percentages of the applications running on the UEs. Based on the formulation in part I, the usage percentage is represented as an application-status differentiation weight, and the $M$ UEs switch between their applications with the usage percentages as equation (\ref{eqn:usageArray}) whose values are taken from the first five rows in Table \ref{table:usage_percentage}. As an illustration, during the simulation time $0 < t < 100$, $\boldsymbol \alpha = \boldsymbol \alpha_1 = \{0.1, 0.5, 0.9, 0.1, 0.5, 0.9, 0.9, 0.5, 0.1, 0.9, 0.5, 0.1\}$, and so forth. It is noteworthy that the temporal application usage percentage $\boldsymbol\alpha_1 = \{0.1, 0.5, 0.9, 0.1, 0.5, 0.9, 0.9, 0.5, 0.1, 0.9, 0.5, 0.1\}$ indicates that the $1^{st}$ utilizes its real-time and delay-tolerant applications $90\%$ and $10\%$ 
of the $0 < t < 100$ simulation time respectively, and likewise for other UEs and applications.

\begin {table}[]
\caption {Applications Status Differentiation}
\label{table:usage_percentage}
\begin{center}
\renewcommand{\arraystretch}{1.4}
\begin{tabular}{| l | l |}
  \hline
  \multicolumn{2}{|c|}{Applications Usage-Percentage} \\  \hline
  $\boldsymbol\alpha_1$ &  \{0.1, 0.5, 0.9, 0.1, 0.5, 0.9, 0.9, 0.5, 0.1, 0.9, 0.5, 0.1\} \\ \hline
  $\boldsymbol\alpha_2$ &  \{0.5, 0.3, 0.2, 0.5, 0.3, 0.2, 0.5, 0.7, 0.8, 0.5, 0.7, 0.8\} \\ \hline
  $\boldsymbol\alpha_3$ &  \{0.5, 0.9, 0.8, 0.5, 0.9, 0.8, 0.5, 0.1, 0.2, 0.5, 0.1, 0.2\} \\ \hline
  $\boldsymbol\alpha_4$ &  \{0.5, 0.3, 0.2, 0.5, 0.3, 0.2, 0.5, 0.7, 0.8, 0.5, 0.7, 0.8\} \\ \hline
  $\boldsymbol\alpha_5$ &  \{0.5, 0.9, 0.8, 0.5, 0.9, 0.8, 0.5, 0.1, 0.2, 0.5, 0.1, 0.2\} \\ \hline
  $\boldsymbol\alpha_a$ &  \{0.1, 0.5, 0.9, 0.1, 0.5, 0.0, 0.9, 0.5, 0.1, 0.9, 0.5, 0.0\} \\ \hline
  $\boldsymbol\alpha_b$ &  \{0.1, 0.5, 0.9, 0.1, 0.5, 0.9, 0.9, 0.5, 0.1, 0.9, 0.5, 0.1\} \\ \hline
\end{tabular}
\end{center}
\end {table}

The sensitivity of the algorithms as for UE rate allocations and bids is depicted in Figure \ref{fig:sim:app_percentage_rates}, where the users rates, e.g. $r_i$ for the $i^{th}$ UE, changes with time for the changing usage percentages in accordance with $\boldsymbol\alpha(t)$; they also converge to optimal rates during each time interval as the resource allocation algorithms execute. Furthermore, Figure \ref{fig:sim:app_percentage_bids} portrays the UE bids, e.g. $w_i$ for the $i^{th}$ UE, as they change during the different time intervals during which application usage percentages alter; however, the bids also converge with time for the changing usage percentages. In particular, we observe that at the commencement of the algorithm the jumps between the bids are larger, whereas these jumps significantly shrink as time proceeds and the rates/bids converge to their optimal values.

\begin{figure}
\centering
\subfigure[UE Rates with Temporally Changing Application Usages.]{\label{fig:sim:app_percentage_rates}\includegraphics[width=\plotwidth]{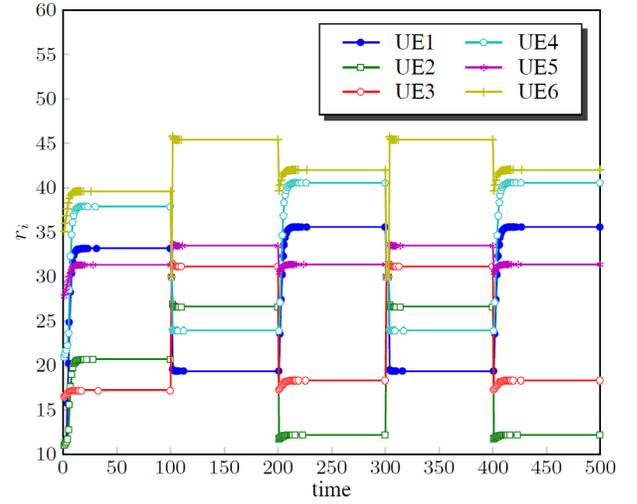}}\qquad
\subfigure[UE Bids with Temporally Changing Applications Usages.]{\label{fig:sim:app_percentage_bids}\includegraphics[width=\plotwidth]{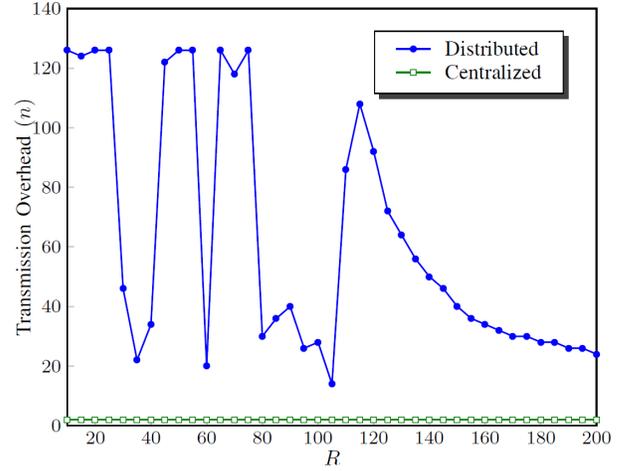}}%
\caption{The UEs allocated rates $r_i$ and pledged bids $w_i$ vary in time as the application usage percentages change. Figure (a) illustrates UE rate changes as the applications usages alters, whereas Figure (b) depicts the UE bid changes as the applications usages vary in time. As time proceeds rate/bid jumps shrink due to the convergence. Furthermore, in every interval rates/bids converge towards optimal ones.}
\end{figure}

In Figure \ref{fig:sim:app_percentage_shadow_price}, we show the changes in the shadow price $p$ vs. time when the application usage percentages change. Here, we can observe a shrinkage of the shadow price variations as the algorithms execute in time. Furthermore, higher usage corresponds to a higher shadow price. Moreover, a comparison between the centralized and distributed algorithms transmission overhead whilst switching the application-usage percentages is illustrated in Figure \ref{fig:sim:app_percentage_tx_overhead}, which shows that the centralized method incurs a significantly lower transmission overhead due to the lesser number of transmissions it needs as opposed to those of the distribution approach. In fact, the centralized scheme transmission overhead plot is very close to the abscissa in view of the drastic less message exhchanges it incurs, and the transmission overhead is independent of the termination threshold $\delta$. On the flip side, the lower termination thresholds augments the 
transmission overhead for the distributed resource allocation in that the smaller thresholds make the algorithm execute more iterations before it converges to optimal rates and the more iterations are tantamount to a larger number of messages exchanged between the eNB and UEs, which in turn escalates the transmission overhead dramatically.

\begin{equation}
\label{eqn:usageArray}
\boldsymbol\alpha(t) = \left\{
  \begin{array}{l l}
    \boldsymbol\alpha_1 & ;\;\;\;0  < t \le 100\\
    \boldsymbol\alpha_2 & ;100 < t \le 200\\
    \boldsymbol\alpha_3 & ;200 < t \le 300\\
    \boldsymbol\alpha_4 & ;300 < t \le 400\\
    \boldsymbol\alpha_5 & ;400 < t \le 500
  \end{array} \right.
\end{equation}

\begin{figure}[t!]
\centering
  \includegraphics[width=\plotwidth]{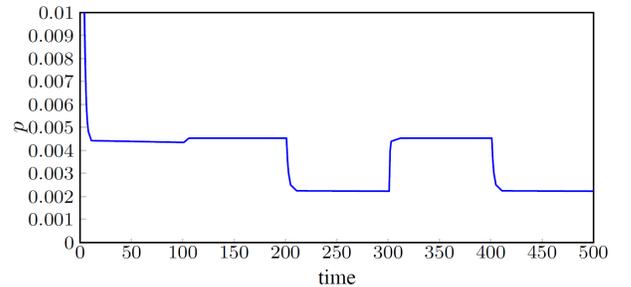}
  \caption{The shadow price varies with temporal changes of the application usage percentages. The more usage, the higher the shadow price.}
  \label{fig:sim:app_percentage_shadow_price}
\end{figure}

Next, section \ref{sec:SensitivityNewUE} discusses the sensitivity of the resource allocation schemes to changes in the number of UEs in the system.

\begin{figure}[t!]
\centering
  \includegraphics[width=\plotwidth]{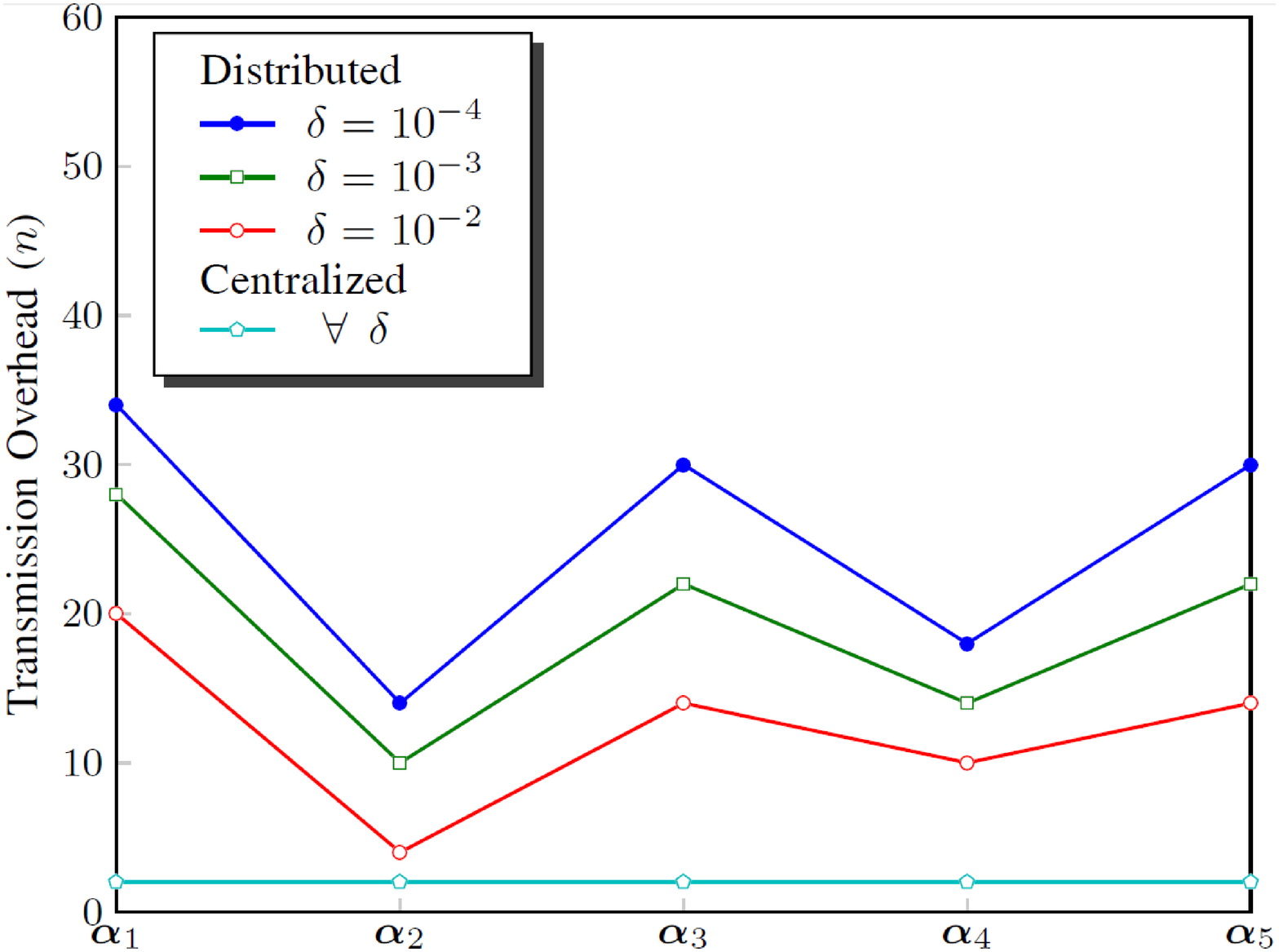}
  \caption{Transmission Overhead: The plots compares the centralized and distributed algorithms with respect to their transmission overhead when application usage percentages change. The centralized algorithm is independent of the termination threshold $\delta$, whereas the smaller thresholds causes more iterations of the algorithm leading to escalated transmission overheads.}
  \label{fig:sim:app_percentage_tx_overhead}
\end{figure}

\subsection{Sensitivity to Changes in the UE Quantity} \label{sec:SensitivityNewUE}
Similarly to the simulations in section \ref{subsec:SensitivityAlpha}, we set $\delta = 10^{-3}$ and $R = 180$ and measure the sensitivity of the proposed algorithms to changes in the number of the UEs in the system, where new users enter the cell so that the UE quantity at time slot $n_1 = 100$ is $M_1 = 5$ and at time instant $n_1 + 1$ the is $M_2 = 6$. Then, the instantaneous number of users, denoted as $M(n)$, is given by the equation (\ref{eqn:no_of_users}).

\begin{equation}
\label{eqn:no_of_users}
M(n) = \left\{
  \begin{array}{l l}
    M_1\:\: ; \:\:n\leq n_1\\
    M_2\:\: ; \:\:n>n_1
  \end{array} \right.
\end{equation}

Furthermore, we set the temporal application status differentiation $\boldsymbol\alpha$ as equation (\ref{eqn:newUserArray}).

\begin{equation}
\label{eqn:newUserArray}
\boldsymbol\alpha(t) = \left\{
  \begin{array}{l l}
    \boldsymbol\alpha_a & ;0   < t \le 100,\\
    \boldsymbol\alpha_b & ;100 < t \le 200.
  \end{array} \right.
\end{equation}

Our experiments show that the distributed algorithm without rebidding as well as the centralized approach deviate from optimal rates, bids, and shadow price under UE quantity changes. Figures \ref{fig:sim:app_percentage_rates} and \ref{fig:sim:app_percentage_bids} respectively show the errors in UE rates and bids, i.e. $|r_i - r_i^{\text{opt}}|$ and $|w_i - w_i^{\text{opt}}|$ for the $i^{th}$ UE, when a new user (here UE6) enters the cell. As we can observe, the centralized algorithm retains its optimal UE rate assignment and bid pledging so that the errors become $0$, whereas the distributed method is concomitant with peak errors before converging to the steady state values. Besides, looking at the different UE rate and bid errors, we observe that the UEs whose real-time applications have higher QoS requirements, tantamount to larger inflection points in their respective sigmoidal utility functions, undergo larger errors vis-a-vis the other UEs.

\begin{figure}
\centering
\subfigure[UE Rates with Temporally Changing Application Usages.]{\label{fig:sim:app_percentage_rates}\includegraphics[width=\plotwidth]{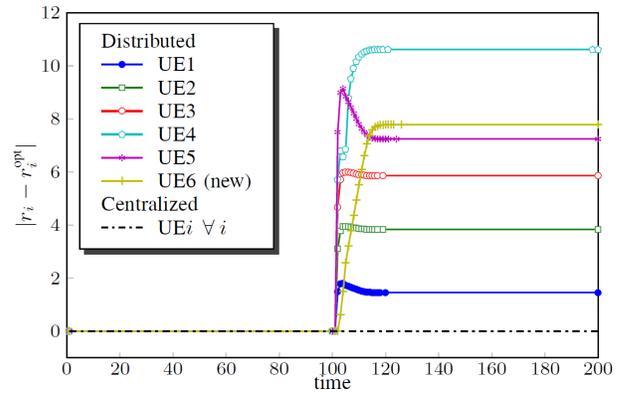}}\qquad
\subfigure[UE Bids with Temporally Changing Applications Usages.]{\label{fig:sim:app_percentage_bids}\includegraphics[width=\plotwidth]{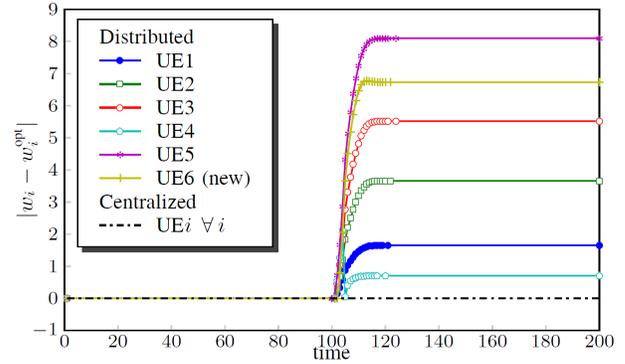}}%
\caption{Comparisons between the centralized and distributed algorithms divergence from optimal solutions (i.e. error) when a new users enters the cellular network. Figure (a) depicts the rate error $r_i$  when a new user $UE6$ enters the network. The dashed line represents the rate error for the one-stage centralized method, whose rate allocations remain optimal (zero error) upon the new UE entry. Figure (b) shows the bid errors when the new UE enters the network, and it shows that the bidding remains optimal for the centralized method. }
\end{figure}

Last, Figure \ref{fig:sim:app_percentage_tx_overhead} depicts the error in the eNB announced shadow price, i.e. $|p - p^{\text{opt}}|$, when UE6 enters the cell. As we can see, the centralized algorithm boasts a zero pricing error, whereas the distributed one is susceptible to about $30\%$ shadow price error.

\begin{figure}[t!]
\centering
  \includegraphics[width=\plotwidth]{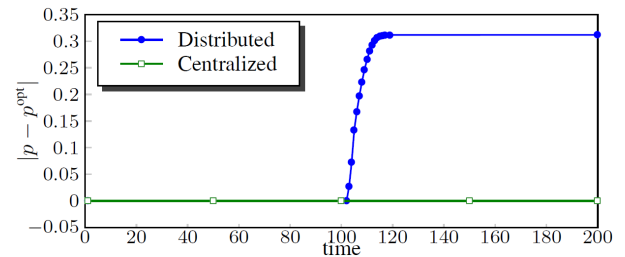}
  \caption{Figure (c) depicts the shadow price error for the distributed and centralized method, where the latter introduces no shadow price errors.}
  \label{fig:sim:app_percentage_tx_overhead}
\end{figure}

\section{Implementation Results} \label{Sec:Implementation}
In this section, we implement the resource allocation developed in Part I on a real-world network. The implementation system model  is depicted in Figure \ref{fig:system_model} and subsumes UEs WiFi-connected to the Internet through a network outfitted with a resource broker (RB) unit, containing the resource allocation code introduced in Part I, installed on a router which shapes the traffic based on application rates assigned by the RB entity in accordance with the resource allocation method in Part I.

\begin{figure}[!htb]
\begin{center}
\includegraphics[width=3.5in]{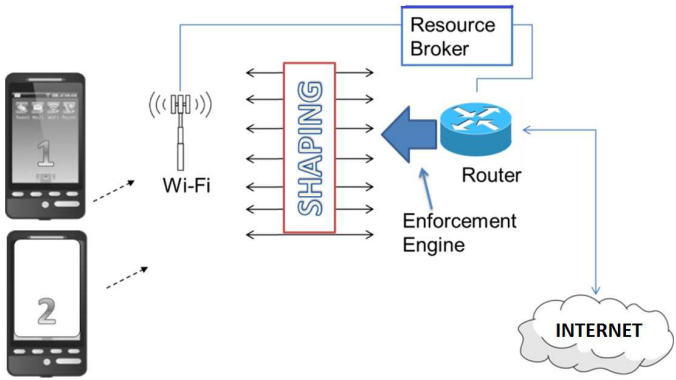}
\end{center}
\caption{Implementation System Model: UEs, WiFi-connected to the Internet, run delay-tolerant and real-time applications whose optimal rates are assigned by the RB unit, which includes the resource allocation method developed in Part I, installed on a router which shapes the traffic based on the rates allocated by the resource allocation method.}\label{fig:system_model}
\end{figure}

In order to implement the scenario in Figure \ref{fig:system_model} on a real-world network, we leverage a personal computer (PC) to configure the network in a distributed manner to decrease the  processing load through the virtual machine (VM) architecture in Figure \ref{fig:VM_TestBed}. Here, we have used a single-socket IBM x3250 M4 server \cite{IBM} with two physical and two Peripheral Component Interconnect (PCI)-enabled ports to create two three-interfaced VMs, between which one hosted the RB entity and the other formed a router including an enforcement engine to manage rate assignments via an onboard Linux router traffic control, and another dedicated file server VM. The two VMs (not the file server one) are annotated as "Guest 1" and "Guest 2" in Figure \ref{fig:VM_TestBed}, where the router and RB sit on Guest 1 and Guest 2, respectively. Besides, we create three virtual switches intended for the phone network, another for office-Internet-connected external devices, and one for network maintenance/
operation issues. The test-bed phones and their wireless access point (WAP) are on their own private network (i.e. the Ethernet network with internet protocol (IP) address $192.168.2.1$), router gateway settings enable connections to the office network (intranet) and Internet, and finally smartphones running Youtube and Hyper Text Transfer Protocol (HTTP) download applications are deployed. The traffic generated by the real-time/delay-tolerant Youtube/HTTP applications is inelastic/elastic, and we apply the resource allocation procedure in Part I to obtain application rates (throughput) and germane subjective QoE, reflected by buffering occurrences for the inelastic Youtube traffic and incomplete downloads for the elastic HTTP traffic. The dearth of the aforesaid events indicates an acceptable QoE for the application users and effectiveness of the resource allocation method in Part I.

\begin{figure}[!htb]
\begin{center}
\includegraphics[width=3.5in]{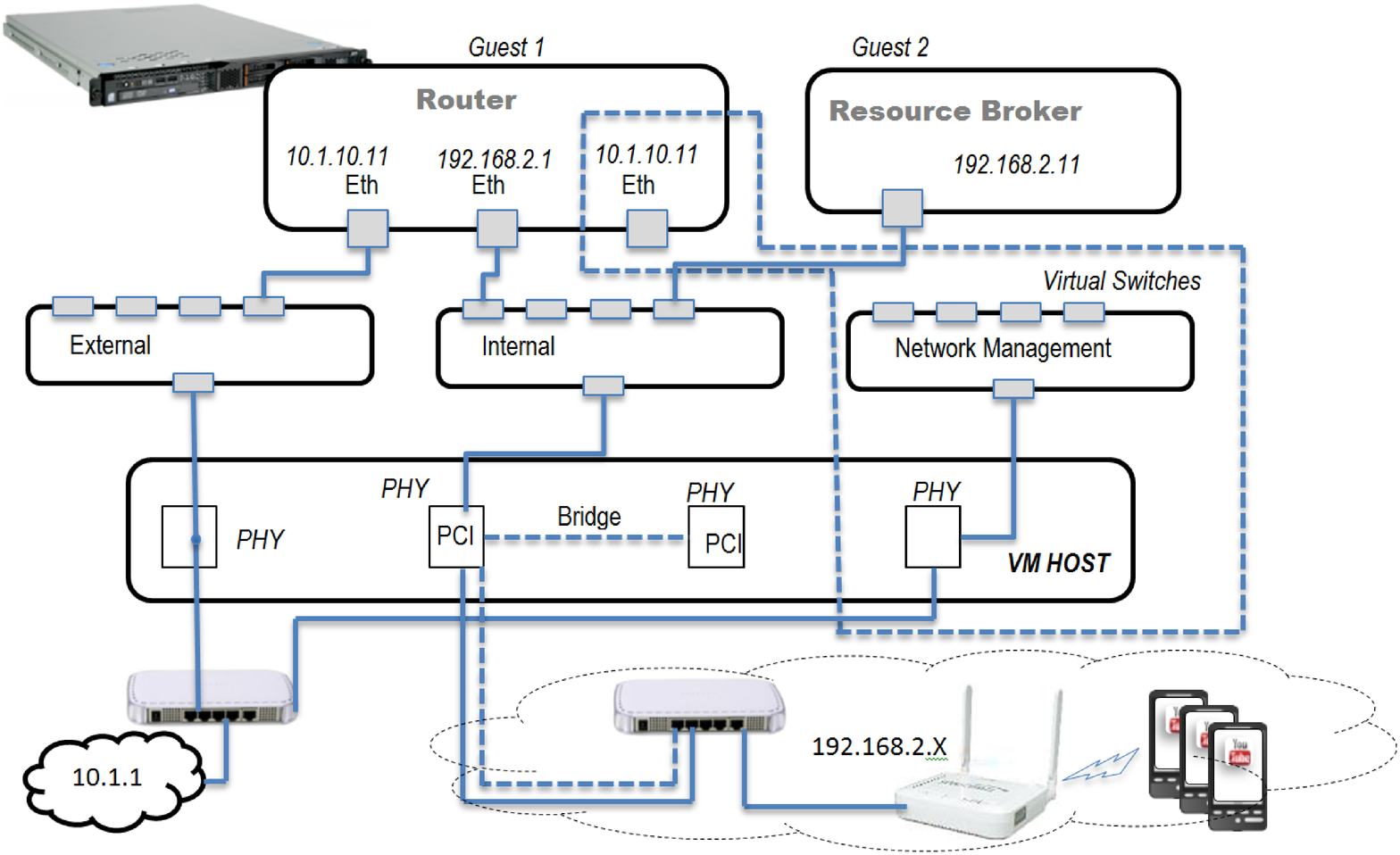}
\end{center}
\caption{Real-World Network Architecture: Two VMs "Guest 1" and "Guest 2" respectively host the router and resource allocation method (Part I) code in the RB unit. Furthermore, an office-Internet virtual switch for external devices, three phone-network virtual switches, and a network management virtual switch are created on a two physical two PCI ported IBM server. The Phones and WiFi WAP have their private network with IP address $192.168.2.1$.}\label{fig:VM_TestBed}
\end{figure}

For the test platform in Figure \ref{fig:VM_TestBed}, the small number of phones working in a high throughput WiFi network is an object of concern as it provides too ideal of a data transfer environment to be able to appropriately illustrate the benefits that may emerge from an RB-induced traffic shaping, i.e. from the resource allocation method in Part I. To observe the traffic shaping effects, we should impose a higher load on the network. We do this simply by restricting the overall network bandwidth to $1$ Mega bits per second (Mbps). In order to make a comparison between the scenarios with and without the resource allocation algorithm of Part I, we first apply the rate assignment scheme to the network in Figure \ref{fig:VM_TestBed} under no bandwidth constraints; then, we introduce an overall network constraint with no traffic shaping, and ultimately contrast non-shaped throughput and QoE observation to a situation where the RB operates under a $1$ Mbps network constraint. The Youtube and HTTP 
applications throughput under neither
network constraints nor the resource allocation application is shown in Table \ref{table:NCnsNShp}, where a low/high rate YouTube application "YouTube 1"/"YouTube 2" and a small/large file download application "HTTP 1"/"HTTP 2" (obtained from content providers of the VM-created file server) run on three UEs in the network whose speed under the absence of applications is measured at $32$ Mbps. For instance, the $1^{st}$ scenario in Table \ref{table:NCnsNShp} indicates that two phones run YouTube 1 and the average bandwidth consumption is $R_{avg} = 3.492$ Mbps, while the single-UE $3^{rd}$ scenario incurs a much lower bandwidth consumption $R_{avg} = 0.951$ Mbps due to the low rate YouTube 1 application. On the other hand, the bandwidth consumption significantly increased to $R_{avg} = 2.11$ Mbps due to the single-UE high rate YouTube 2 application in the $4^{th}$ scenario, and a similar rate rise is apparent in the transition from the $5^{th}$ scenario to the $6^{th}$ one where the high rate HTTP 2 replaces 
the low rate HTTP 1 at the $2^{nd}$ phone amounting to $24.866$ Mbps bandwidth consumption. The last scenario is concomitant with a lower
rate $R_{avg} = 13.747$ Mbps as opposed to the $6^{th}$ scenario ($R_{avg} = 24.866$ Mbps) in spite of augmenting a YouTube 1 application to phone 3 in the latter case. This can be explained by the need to transfer more bits which happen over a longer time interval vis-a-vis that of the $6^{th}$ case, thereby the throughput slashes down relative to the $6^{th}$ configuration.

\begin {table}[]
\caption {Network Throughput in Mbps for the Network in Figure \ref{fig:VM_TestBed} under neither Bandwidth Constraints nor Traffic Shaping.}
\label{table:NCnsNShp}
\begin{center}
\renewcommand{\arraystretch}{1.4}
\begin{tabular}{| l || l | l | l | l |}
  \hline
  Scenario & Phone 1   & Phone 2 & Phone 3   & $R_{avg}$ Mbps \\  \hline\hline
  1 & Youtube  1 & Youtube 2 &   -       & 3.492   \\ \hline
  2 & Youtube 1 & HTTP 1    &   -       & 4.050   \\ \hline
  3 & Youtube 1 & -       &   -       & 0.951   \\ \hline
  4 & Youtube 2 & -       &   -       & 2.11    \\ \hline
  5 & HTTP 1    & HTTP 1  &   -       & 4.262   \\ \hline
  6 & HTTP 1    & HTTP 2  &   -       & 24.866  \\ \hline
  7 & HTTP 1    & HTTP 1  & Youtube 1 & 13.747 \\ \hline
  \multicolumn{5}{|c|}{Network Speed with neither Application nor Rate Constraints: $32$ Mbps} \\  \hline
\end{tabular}
\end{center}
\end {table}

To see the effect of the resource allocation developed in Part I and implemented in the RB entity, we throttle the overall network bandwidth to $R = 1$ Mbps. Without loss of generality, we focus on deploying only two phones in the network in Figure \ref{fig:VM_TestBed} with IP addresses and a YouTube 1 and HTTP 1 application as illustrated in Table \ref{table:PhnIP}. Running the experiment with no resource allocation applied, we observe that Youtube 1 incurred multiple buffering periods indicating a poor QoE from its user's perspective, the overall average bandwidth usage was $0.963$ Mbps, and HTTP 1 download completed in $1200$ seconds (s) (\ref{table:Performance}). Using WireShark \cite{Chappell2010}, we obtain the traces in Figure \ref{fig:NoCsNoSh}, where both applications sharing the total $1$ Mbps bandwidth annotated on the black curve alternate in bursty transmission intervals. In particular, at certain times, the HTTP 1 application (the red curve) utilizes the entire available bandwidth, shown by the 
red curve reaching the black curve, which simultaneously zeros the Youtube 1 throughput illustrated by the green curve lowering to the abscissa (time axis). This behavior adversely affects the QoE for the UE 1.

\begin {table}[]
\caption {Smartphone Applications under $R = 1$ Mbps Network Constraint with/without Traffic Shaping.}
\label{table:PhnIP}
\begin{center}
\renewcommand{\arraystretch}{1.4}
\begin{tabular}{| l || l | l | l |}
  \hline
  Phone & IP & Traffic Type & Application \\  \hline \hline
  UE 1 & $192.168.2.57$ & Inelastic & YouTube 1\\ \hline
  UE 2 & $192.168.2.98$ & Elastic & HTTP 1\\ \hline
\end{tabular}
\end{center}
\end {table}

\begin{figure}[!htb]
\begin{center}
\includegraphics[width=3.5in]{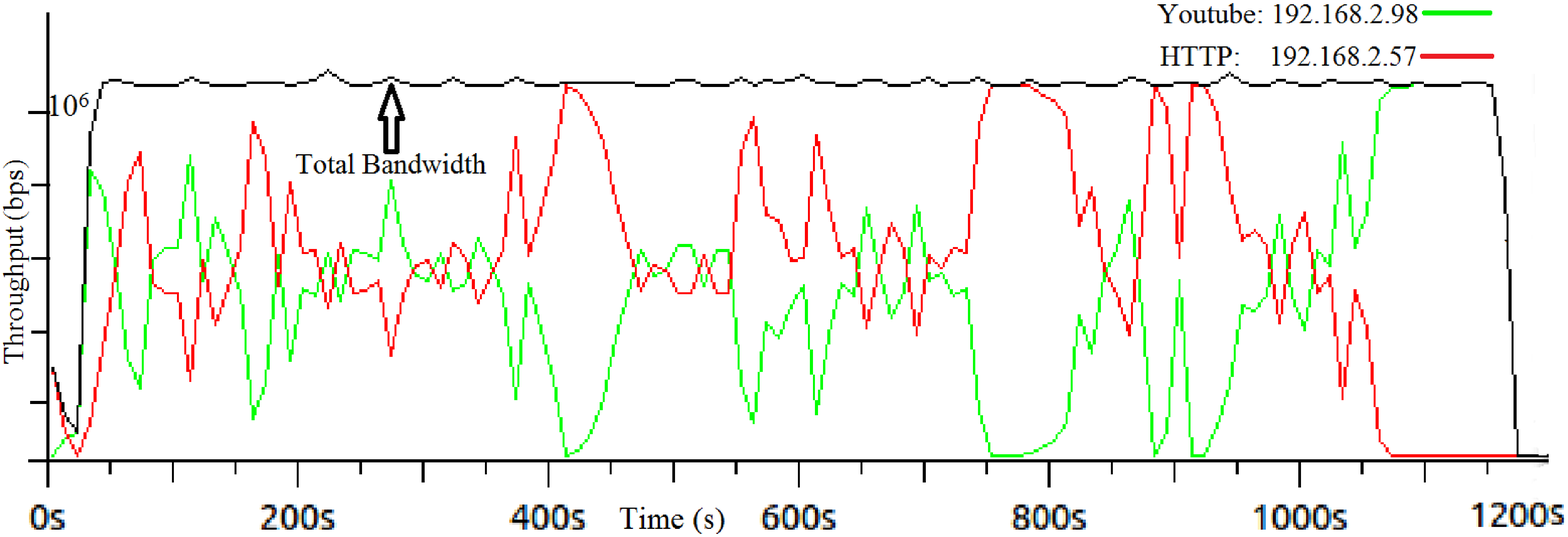}
\end{center}
\caption{Wireshark Throughput Analysis without Applying the Resource Allocation in Part I: For the network in Figure \ref{fig:VM_TestBed} with the parameters in Table \ref{table:PhnIP}  under an overall bandwidth constraint $R = 1$ Mbps and no traffic shaping, black/green/red curve shows the network/YouTube 1/HTTP 1 throughput. HTTP 1 downloaded in $1200$ s with YouTube 1 incurring multiple buffering, adversely affecting its QoE as HTTP 1 (red curve) uses the entire available capacity (hits the black curve).}\label{fig:NoCsNoSh}
\end{figure}

The same situation of the network in Figure \ref{fig:VM_TestBed} with the parameters in Table \ref{table:PhnIP} is repeated when the resource allocation method in Part I is applied to have the router shape the applications traffic by optimally assigning them rates based on their requirements. To apply the resource allocation method developed in Part I, the phones UE 1 and UE 2 register with the RB entity, where the rate allocation code runs and calculates rates to be enforced at the router. The average bit rates for the YouTube 1 and HTTP 1 are respectively $731$ and $267$ kbps, the convergence time for the algorithm is measured at $528$ ms, and the overall throughput becomes $0.758$ Mbps which is less than the maximum $1$ Mbps available capacity due to periods over which no YouTube traffic load is present on the network. Using Wireshark, the rates of the YouTube 1 and HTTP 1 applications when the resource allocation is leveraged in the network under overall $R = 1$ Mbps constraint are depicted in Figure \
ref{fig:NoCsSh}.

Here, the black curve indicates the overall available bandwidth and the ordinate (throughput axis) shows that the network bandwidth is confined to $1$ Mbps. As we can see, YouTube 1 (green curve) consumes more resources per the rate allocation interval than does the HTTP 1 application, whose download time expectedly takes longer to be completed at $2650$ s ((\ref{table:Performance})). Interestingly, there are intervals where YouTube 1 rate becomes $0$, over which HTTP 1 obtains more bandwidth; this zero-grounding behavior is analogously observed for the HTTP 1 download. In this experiment, we observe no YouTube buffering occurrences, thereby it provisions a better video watching experience for the user as opposed to the unshaped traffic scenario depicted in Figure \ref{fig:NoCsNoSh}. Such a dearth of buffering speaks directly to the speculated QoE in that the real-time YouTube 1 application is provided with a consistent rate assignment such that it is able to fill the buffer and does not require any more 
bandwidth usage.

\begin{figure}[!htb]
\begin{center}
\includegraphics[width=3.5in]{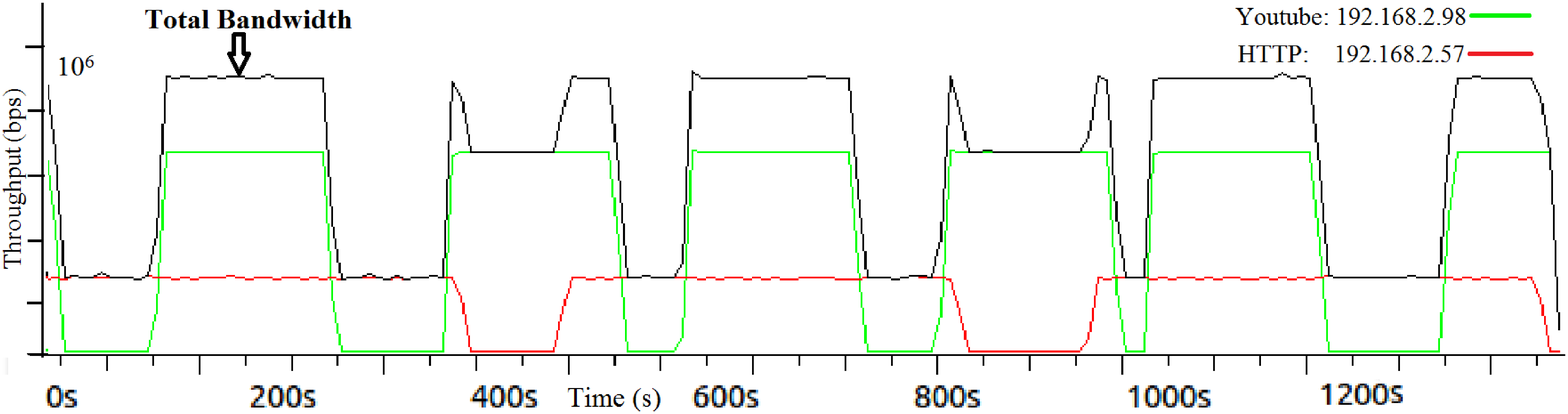}
\end{center}
\caption{Wireshark Throughput Analysis with Applying the Resource Allocation in Part I: For the network in Figure \ref{fig:VM_TestBed} with the parameters in Table  \ref{table:PhnIP} under bandwidth $R = 1$ Mbps and shaping, black/green/red curve shows the network/YouTube 1/HTTP 1 throughput. HTTP 1 downloaded in $2650$ s with no YouTube 1 buffering occurances. Occasionally, one application rate zeros and the other application utilizes the maximum bandwidth to achieve a consistent allocation to elevate users' QoE.}\label{fig:NoCsSh}
\end{figure}

As we observed in this real-world implementation of the resource allocation algorithm presented in Part I of this paper, the resource allocation elevates the QoE despite the fact that less resources are consumed. This directly helps communications carriers to reduce their operation expenditure (OPEX) and customer churn. As a case in point, the algorithm-induced traffic shaping decreased the bandwidth consumption from $0.963$ Mbps to $0.758$ Mbps, thereby a $0.205$ Mbps reduction of resource consumption is stemmed which has desirable monetary indications by using up less resources without degrading users' QoE. As such, an QoE remedial measurement is implied in light of employing the developed resource allocation scheme.

\begin {table}[]
\caption {Bandwidth Consumption and Download Time.}
\label{table:Performance}
\begin{center}
\renewcommand{\arraystretch}{1.4}
\begin{tabular}{| l || l | l |}
  \hline
  Performance & Shaped Traffic & Unshaped Traffic \\  \hline \hline
  HTTP 1 Download Time (s) & $1200$ & $2650$\\ \hline
  Youtube 1 Buffering & No & Yes\\ \hline
  Total Bandwidth (kbps) & $758$ & $951$\\ \hline
\end{tabular}
\end{center}
\end {table}

\section{Conclusion}\label{sec:conclude}
In part I of this paper, we introduced a novel QoS-minded centralized and a distributed algorithm for the resource allocation within the cells of cellular communications systems. We formulated the centralized and distributed approaches as respectively a singular and double utility proportional fairness optimization problems, where the former allocated running applications rates directly by the allocation entity such as an eNB in response to the UE utility parameters sent, whereas the latter assigned the UE rates by the eNB in its first stage followed by the application rate allocation by the UEs in its second stage. Users ran both delay-tolerant and real-time applications mathematically modelled correspondingly as logarithmic and sigmoidal utility functions, where the function values represented the applications QoS percentage. Both of the proposed resource allocation formulations incorporated the service differentiation, application status differentiation modelling the applications usage percentage, and 
subscriber differentiations amongst subscribers priority within networks into their formulation. Not only did we prove that the proposed resource allocation problems were convex and solved them through Lagrangian of their dual problems, but also we proved the optimality of the rate assignments and the mathematical equivalence of the proposed distributed and centralized resource allocation schemes.

Furthermore, we analyzed the algorithm convergence under varying sums of resources available to the eNB and introduced robustness into the distributed algorithm by incorporating decay functions into the aforementioned algorithm so that it converged to optimal rates for both high and low traffic loads occurring during the day by damping the rate assignment fluctuations resulting from scarcity of resources that appeared particularly in peak-traffic circumstances.

In part II, we proved that the centralized resource allocation algorithm incurs the minimum possible transmission overhead during the rate assignment procedure. Moreover, we derived the transmission overhead based on changes in the number of UEs in the cell under the existence of the dearth a rebidding policy. We proved that the absence of a rebidding policy entails a transmission overhead less that or equal to that of the one with presence of a rebidding policy. In addition, we proved that the rate assignments remain optimal for variations in the network's UE quantity for the distributed resource allocation approach under a rebidding policy, whilst the rates are not optimal when the number of UEs change for the distributed approach with a no-rebidding policy and also for the centralized resource allocation scheme. Besides, we proved the rate assignments remain optimal when the application usage percentages changes in case of the centralized resource allocation scheme, but it is not optimal when the 
aforesaid parameter varies in case of the distributed resource allocation method. We also analyzed the transmission overhead and its changes for both the centralized and distributed method when the application usage percentages vary in the network.

Furthermore, we provisioned a traffic-dependent pricing technique for network providers using which they could flatten the traffic loads during peak traffic hours. The devised pricing could be as well employed by network subscribers to potentially decrease the cost of utilizing the network by using services at low-cost (low-traffic) time periods. Ultimately, through simulations, we showed that the rate assignment, bid pledging, and pricing errors for the centralized approach is $0$ while those of the distributed approach converge to their steady state values different from the optimal ones.

Ultimately, We set up a real-world network to evaluate the QoE users perceive when the procedure is implemented in practice. The large-scale network configuration included Youtube and HTTP applications on UEs connected to the Internet through their WAP and a router which ran the resource allocation routine and enforced the rates. We realized that the absence of the resource allocation algorithm in the network caused multiple buffering instances of the real-time Youtube application, thereby the users QoE was adversely undermined. On the other hand, the presence of the algorithm, through which the rates were assigned to the applications and enforced at the gateway router, eliminated any Youtube buffering at the expense of lengthening the duration of delay-tolerant download applications. Therefore, applying the resource allocation method significantly escalated the users QoE without hurting QoS requirements of applications. Finally, we realized that despite the dearth of Youtube buffering under the presence of 
the resource allocation algorithm,
the applications consumed less resources as opposed to the QoE-hurting algorithm-absent scenario. Consequently, the bandwidth conservation yields in a lower OPEX for the network as less to-be-paid resources are consumed.

\bibliographystyle{ieeetr}
\bibliography{pubs}
\end{document}